\documentclass[sts]{imsart}

\RequirePackage{amsthm,amsmath,amsfonts,amssymb}
\usepackage{thmtools}
\renewcommand\thmcontinues[1]{continued}
\RequirePackage{natbib}

\RequirePackage[colorlinks,breaklinks=true,citecolor=blue,urlcolor=blue,linkcolor=blue]{hyperref}
\RequirePackage{graphicx}

\usepackage{comment}
\usepackage{lineno}
\usepackage[enable]{easy-todo}
\usepackage{amsmath, amsthm, amssymb, bigints, bm}
\usepackage[plain]{algorithm2e}
\usepackage{rotating}
\usepackage{tikz}
\usetikzlibrary{arrows, decorations.pathreplacing, positioning, shapes,arrows.meta}
\usepackage{chngcntr}
\usepackage{apptools}
\AtAppendix{\counterwithin{lemma}{section}}
\AtAppendix{\counterwithin{proposition}{section}}
\AtAppendix{\counterwithin{theorem}{section}}
\AtAppendix{\counterwithin{corollary}{section}}
\AtAppendix{\counterwithin{definition}{section}}
\AtAppendix{\counterwithin{figure}{section}}
\AtAppendix{\counterwithin{table}{section}}
\AtAppendix{\counterwithin{algocf}{section}}
\usepackage[capitalise]{cleveref}
\crefname{assumption}{Assumption}{Assumptions}
\usepackage{booktabs, multirow}
\usepackage{enumerate}

\startlocaldefs

\usepackage{mathtools}

\definecolor{blue-violet}{rgb}{0.54, 0.17, 0.89}
\definecolor{antiquefuchsia}{rgb}{0.57, 0.36, 0.51}
\definecolor{amethyst}{rgb}{0.6, 0.4, 0.8}
\definecolor{blue-violet}{rgb}{0.54, 0.17, 0.89}
\definecolor{ao}{rgb}{0.0, 0.5, 0.0}
\definecolor{blue(ncs)}{rgb}{0.0, 0.53, 0.74}
\definecolor{dgreen}{rgb}{0.12, 0.3, 0.17}
\definecolor{cadmiumgreen}{rgb}{0.0, 0.42, 0.24}
\definecolor{darkolivegreen}{rgb}{0.33, 0.42, 0.18}
\definecolor{dartmouthgreen}{rgb}{0.05, 0.5, 0.06}

\newcommand{\tild}{\raise.17ex\hbox{ $\scriptstyle\sim$ }}

\newcommand{\indep}{\mathrel{\text{\scalebox{1.07}{$\perp\mkern-10mu\perp$}}}}

\DeclareMathOperator{\E}{\mathbb{E}}

\DeclareMathOperator{\De}{De}

\DeclareMathOperator{\nd}{Nd}

\DeclareMathOperator{\An}{An}

\DeclareMathOperator{\Pa}{Pa}

 \theoremstyle{plain}
\newtheorem{theorem}{Theorem}

\newtheorem{proposition}{Proposition}
\newtheorem{lemma}{Lemma}

\newtheorem{assumption}{Assumption}
\theoremstyle{definition}
\newtheorem{definition}{Definition}

\newtheorem{example}{Example}
\theoremstyle{remark}
\newtheorem{remark}{Remark}

\let \tilde \widetilde
\let \epsilon \varepsilon

\newcommand{\g}{\mathcal{G}}
\newcommand{\Ccommon}{C_{\mathrm{conj.cause}}}
\newcommand{\Cpre}{C_{\mathrm{pre.treatment}}}
\newcommand{\CvdW}{C_{\mathrm{disj.cause}}}
\newcommand{\Cdisj}{C_{\cup}}
\newcommand{\Cconj}{C_{\cap}}
\newcommand{\AY}{\text{AY}}
\newcommand{\YA}{\text{YA}}
\newcommand{\CAY}{C_{\AY}}
\newcommand{\CYA}{C_{\YA}}

\newcommand{\predA}{\mathfrak{R}_{A}}
\newcommand{\predYA}{\mathfrak{R}_{Y|A}}
\newcommand{\PredA}{R_{A}}
\newcommand{\PredYA}{R_{Y|A}}

\endlocaldefs

\begin{document}

\begin{frontmatter}

  \title{Confounder Selection: Objectives and Approaches}
  \runtitle{Confounder Selection}

  \begin{aug}
    \author[A]{\fnms{F.~Richard}~\snm{Guo}\ead[label=e1]{ricguo@statslab.cam.ac.uk}},
    \author[B]{\fnms{Anton~Rask}~\snm{Lundborg}\ead[label=e2]{arl@math.ku.dk}}
    \and
    \author[C]{\fnms{Qingyuan}~\snm{Zhao}\ead[label=e3]{qyzhao@statslab.cam.ac.uk}}

    \address[A]{F. Richard Guo is Research Associate, Statistical Laboratory,
    University of Cambridge, Cambridge, England\printead[presep={\ }]{e1}.}

    \address[B]{Anton Rask Lundborg is Postdoc, Department of Mathematical Sciences,
    University of Copenhagen, Copenhagen, Denmark\printead[presep={\ }]{e2}.}

    \address[C]{Qingyuan Zhao is Assistant Professor, Statistical Laboratory,
    University of Cambridge, Cambridge, England\printead[presep={\ }]{e3}.}

    \end{aug}

  \begin{abstract}
    Confounder selection is perhaps the most important step in
    the design of observational studies. A number of criteria, often
    with different objectives and approaches, have been proposed, and
    their validity and practical value have been debated in the
    literature. Here, we provide a unified review of these criteria and
    the assumptions behind them. We list several objectives that
    confounder selection methods aim to achieve and discuss the amount
    of structural knowledge required by different approaches. Finally,
    we discuss limitations of the existing approaches and implications
    for practitioners.
  \end{abstract}


\end{frontmatter}

\section{Introduction}
\label{sec:introduction}

When designing an observational study to estimate the causal effect of
a treatment on an outcome, possibly the most important task
is to select covariates that should be measured and controlled
for. A variety of criteria and methods have been proposed, sometimes
as formal algorithms but most often as loosely stated principles. These
proposals differ quite substantially in their objectives and
approaches, so it is not uncommon that methods developed for one
objective fail in achieving another, thus causing a great deal
of confusion for methodologists and practitioners.

To get a sense of the variety of criteria for confounder selection in
the literature, a brief literature review is helpful. In a highly cited
tutorial of propensity score
methods, \citet[p.\ 414]{austin2011} suggested
that there are four possible sets of confounders to control for: ``all measured
baseline covariates, all baseline covariates that are associated with treatment
assignment, all covariates that affect the outcome (i.e., the
potential confounders), and all covariates that affect both treatment
assignment and the outcome (i.e., the true confounders).'' Citing
previous simulation studies
\citep{austin07_compar_abilit_differ_propen_score,brookhart06_variab_selec_propen_score_model},
Austin concluded that ``there were merits to including only the
potential confounders or the true confounders in the propensity score
model''. The idea that only controlling for the ``true confounders''
is sufficient, if not superior, is commonplace in practice
\citep{glymour08_method_chall_causal_resear_racial} and methodological
development
\citep{ertefaie17_variab_selec_causal_infer_using,shortreed17_outcom_lasso,koch20_variab_selec_estim_causal_infer}.
We will call this the ``conjunction heuristic''; other authors have referred to it as the
  ``common cause principle''. It is well known that
this approach can select too few covariates if some confounders are
not observed \citep{vanderweele2019principles}.

Another widely used criterion for confounder selection is to simply
use all observed pre-treatment covariates. This is primarily driven by
the heuristic that an observational study should be designed
to emulate a randomized experiment, as the latter stochastically
balances all pre-treatment covariates, observed or unobserved. This
consideration led \citet[p.\ 1421]{rubin2009should} to conclude that ``I
cannot think of a credible real-life situation where I would
intentionally allow substantially different observed distributions of a
true covariate in the treatment and control groups.'' There, Rubin was
defending this ``pre-treatment heuristic'' against counterexamples
stemming from the graphical theory of causality
\citep{shrier2008,pearl2009remarks,sjolander2009}. In these
counterexamples, conditioning on some pre-treatment covariates
introduces ``collider bias'' to the causal analysis because
conditioning on a variable can induce dependence between its parents
in a causal graph; see \cref{fig:sel}(a) below. This is often seen as
a dispute in the long-standing friction between the graphical and the potential
outcome (or counterfactual) approaches towards causal inference: the
mathematical counterexample seems to make it easy for the graphical approach
to claim a victory. However, the target trial emulation perspective
may still be very useful in practice
\citep{hernan2016using,hernan2020causal}.

In response to the drawbacks of the common cause and pre-treatment
heuristics, \citet{vanderweele2011new} proposed the ``disjunctive cause
criterion'' that selects pre-treatment covariates that are causes of
the treatment, the outcome, or both (throughout this article, causes
include both direct and indirect causes). They proved a remarkable
property that the proposed subset is sufficient to control for
confounding if the observed covariates contains at least one subset
that is sufficient to control for confounding. This criterion was
later revised in \citet{vanderweele2019principles} to additionally
account for (1) instrumental variables, which tend to reduce
efficiency under the causal model assumption or amplify bias in
misspecified models, and (2) proxies of unmeasured confounders, which
tend to reduce bias when controlled for.

Another approach for confounder selection in causal inference is statistical variable
selection: rather than using substantive knowledge
about the causes of the treatment and the outcome, can we select
confounders using statistical tests (usually of conditional
independence) with observational data? Besides the methods mentioned previously
that use the conjunction heuristic,
\citet[Theorem 4.3]{robins97_causal_infer_compl_longit_data} implied that one can
sequentially discard covariates that are irrelevant to determining the
treatment or the outcome, which may be viewed as another instance of the conjunction heuristic; see also
\citet[\S3.3]{greenland99_confoun_collap_causal_infer} and
\citet[\S7.4]{hernan2020causal}. Later, \citet{de2011covariate}
extended this approach and devised an iterative algorithm that
converges to a minimal subset of covariates that suffice to control for confounding; see \cref{sec:no-structural} and also
\cite{persson17_data_driven_algor_dimen_reduc_causal_infer}. Others
have proposed a more global approach that attempts to learn the
causal graph (or its Markov equivalence class) first and then use the
learned graph to determine the set of covariates to control for
\citep{maathuis2015generalized,haeggstroem17_data_driven_confoun_selec_via}; however,
the learned graph class typically does not yield a unique set of
covariates to control for. A statistical version of the disjunctive cause criterion, albeit
motivated by some completely different considerations, can be found in
\citet{belloni13_infer_treat_effec_after_selec}. There, the authors
showed that using the union of variables selected in the treatment
model and the outcome model is critical to coping with incorrectly
selected variables in high-dimensional problems.

More recently, an emerging literature considers the problem of
selecting the optimal set of covariates to control for that maximizes
the statistical efficiency in estimating the causal effect of interest \citep{kuroki2003covariate,henckel2022,rotnitzky2020efficient,guo2022variable}.
These algorithms require knowing the structure of the underlying causal graph.

At this point, it should have become abundantly clear that there are
many different objectives and approaches in confounder selection,
which often lead to conflicting views and criteria. In the rest of
this article, we will give a more in-depth discussion of these
objectives (\Cref{sec:object-conf-select}) and existing approaches
(\Cref{sec:appr-conf-select}), in hope that a clear description of
the considerations in the literature will reveal their merits and
limitations.
We conclude the
article with some further discussion in
\Cref{sec:discussion}. Technical background and proofs are deferred to
the Appendix. Throughout the paper, random vectors are sometimes
viewed as sets, allowing us to use set operations to describe
confounder selection.

\section{Objectives of confounder selection}
\label{sec:object-conf-select}

We consider the canonical problem of
estimating the causal effect of a binary treatment $A$ on an outcome
$Y$. For $a=0,1$, let
$Y_a$ denote the potential outcome had the treatment been administered
at level $a$. Let $S$ be the set of pre-treatment covariates that are
already measured or can
be potentially measured. To be practical, we require any
selected set of confounders to be a subset of $S$.

We use $Z$ to denote the superset of $S$ that consists of all the pre-treatment covariates,
observable or not, that are relevant to estimating the
causal effect of $A$ on $Y$.
Formally, we assume that $Z \cup \{A,Y\}$ is \emph{causally closed} in the
sense that every non-trivial common cause of two distinct variables in
$Z \cup \{A, Y\}$ also belongs to the set itself. Here, a common cause of
two variables is said to be non-trivial if it affects either variable directly, i.e.,
not only through the other common causes of the two variables.
Without this restriction, the set $Z$ may be too large because any
cause of a common cause of $A$ and $Y$ is still a common cause of $A$
and $Y$. For a concrete choice of $Z$, we can take it
to be the smallest superset of $S$ such that $Z \cup \{A,Y\}$ is causally closed, given by
\begin{equation} \label{eqs:Z-by-closure}
Z = \overline{S \cup \{A,Y\}} \setminus \{A,Y\},
\end{equation}
where $\overline{S \cup \{A,Y\}}$ is the \emph{causal closure} of $S \cup \{A,Y\}$.
It can be shown that this choice of $Z$ is pre-treatment (\cref{prop:Z-pre}).
The reader is referred to \cref{apx:closure} for a precise definition
of causal closure using causal graphical models.

\subsection{Primary objectives}

Naturally, the primary objective of confounder selection is
to identify a subset of covariates $C \subseteq S$ that suffices to
control for confounding. That is, we would like
to select a subset $C$ such that every potential outcome is
independent of the
treatment in every stratum defined by $C$.
\begin{definition} \label{def:control}
  A set $C$ of variables is said to \emph{control for confounding} or
  be a \emph{sufficient adjustment set} if
  \begin{equation} \label{eq:no-confounding}
  Y_a \indep A \mid C \quad \text{for } a=0,1.
\end{equation}
\end{definition}
The condition \eqref{eq:no-confounding} is known as weak conditional ignorability
\citep{rosenbaum83_centr_role_propen_score_obser,greenland2009identifiability}
or conditional exchangeability \citep[\S2.2]{hernan2020causal} in the
literature. Depending on the problem, there may be no, one, or
multiple sufficient adjustment sets.

When a sufficient adjustment set $C$ exists, under the positivity or overlap
assumption that $0 < P(A=1 \mid C) < 1$ holds with probability one, we can identify
the marginal distribution of the potential outcomes by ``controlling''
or ``adjusting'' for $C$ in the following way:
\begin{equation} \label{eq:confounder-adjustment}
  \mathbb{P}(Y_a = y) = \E \left[ \mathbb{P}(Y=y \mid A=a,
    C)\right], \quad a=0,1.
\end{equation}
This is also known as the
back-door formula \citep{pearl1993bayesian}.
Equation \eqref{eq:confounder-adjustment} assumes the outcome $Y$ is
discrete and can be easily extended to the continuous case. It
follows that functionals defined by the distributions of $Y_a$ for $a=0,1$
(such as the average treatment effect $\E[Y_1] - \E[Y_0]$) can be
identified from the observed distribution of $(A,C,Y)$. 

In this article, we will focus on
confounder adjustment in the sense given by
\eqref{eq:confounder-adjustment}, as this is
by far the most popular approach in practice. Before proceeding, we
mention some alternative approaches when no such set $C$
exists. Besides using the back-door formula
\eqref{eq:confounder-adjustment}, the causal effect may be
identified using the front-door formula \citep{pearl1995causal} and
other graph-based formulae
\citep{tian2002general,shpitser2006identification,shpitser2022multivariate}. With
additional assumptions, causal identification may be possible by using
instrumental variables
\citep{imbens14_instr_variab,wang17_bound_effic_multip_robus_estim}
and proxies of the confounders
\citep{miao18_ident_causal_effec_with_proxy,tchetgen2020introduction}. Note that the last
approach requires using a different adjustment formula for proxy confounders; a different strand of literature has
investigated the bias-reducing effect of adjusting for proxy
confounders via \eqref{eq:confounder-adjustment}
\citep{greenland80_the_effec_of_miscl_in,ogburn12_nondif_miscl_binar_confoun,ogburn12_bias_atten_resul_nondif_mismeas}.

Much of the disagreement and confusion about confounder selection can be attributed
to not distinguishing the observable set of pre-treatment covariates $S$ from the
full set of pre-treatment covariates $Z$, or more broadly speaking,
not distinguishing the design of observational studies from the
analysis of observational studies.
In fact, the common cause and the pre-treatment heuristics
both achieve the primary objective if $S = Z$ (because we assume $Z
\cup \{A,Y\}$ is causally closed); see
\Cref{apx:graphical}. This may explain why these criteria are often
disputed when different parties make different assumptions about $S$
and $Z$.

\subsection{Secondary objectives}

Besides controlling for confounding, there are many other
objectives for confounder selection. An incomplete list includes
\begin{enumerate}
\item dimension reduction for better transparency, interpretability,
  accuracy or stability
\citep{de2011covariate,vansteelandt2012model,belloni13_infer_treat_effec_after_selec,greenland2016outcome,loh2021confounder},

\item robustness against misspecified modeling assumptions, especially
  when a posited confounder can be an instrumental variable
\citep{bhattacharya2007instrumental,wooldridge16_shoul_instr_variab_be_used,myers11_effec_adjus_instr_variab_bias,ding17_instr_variab_as_bias_amplif},

\item statistical efficiency in estimating the causal effect \citep{kuroki2003covariate,brookhart06_variab_selec_propen_score_model,henckel2022,rotnitzky2020efficient,guo2022variable,tang2022ultra}, and

\item ethical and economic costs of data collection
  \citep{smucler2022note}.
\end{enumerate}

Generally speaking, these considerations are not as crucial as the
primary objective of controlling for confounding bias. Nonetheless,
they are still very important and may decide whether an
observational study is successful. In the following section, we will give a review
of various approaches to confounder selection. As these methods are
usually developed to achieve or optimize just one objective,
practitioners must balance achieving the different objectives for their problem at hand.

\section{Approaches to confounder selection}
\label{sec:appr-conf-select}

As suggested by \citet{vanderweele2019principles}, confounder
selection methods can be based on substantive knowledge or
statistical analysis. Substantive knowledge of the causal structure is usually
represented by a causal directed acyclic graph (DAG)
$\g$ over the vertex set $\{A,Y\} \cup Z$.\footnote{This causal graph $\g$
  can be the \emph{latent projection} \citep{verma1990equivalence} of a larger ground causal DAG $\bar{\g}$ onto $\{A,Y\} \cup Z$. Because we assume that $\{A,Y\}\cup
  Z$ is causally closed with respect to the underlying $\bar{\g}$ (see \cref{apx:closure}), the latent projection $\g$ does not contain bidirected edges and is
  hence a DAG.} A directed edge in the graph such as $u \rightarrow v$
represents a direct causal effect from $u$ to $v$.
Let $\An_{\g}(A)$ be the set of ancestors of $A$ in $\g$, including $A$ itself.
That is, $\An_{\g}(A)$ consists of $A$ and all the variables in $\g$ that have a causal path to
$A$.

Given $\g$, we assume that the distribution of the potential outcomes obeys the single-world
intervention graphs derived from $\g$ \citep{richardson2013single}, which implies that the distribution of the observed variables $\{A,Y\} \cup Z$
obeys the Bayesian network model represented by $\g$
\citep{pearl1988book}. The conditional
independencies imposed by the model can then be read off from $\g$ with the
d-separation criterion. In the context of confounder
selection, a useful concept is the \emph{faithfulness} assumption, which posits the reverse:
conditional independence in the observed distribution implies the
corresponding d-separation in $\g$. A short primer on causal graphical
models can be found in \cref{apx:graphical}.

We shall always assume that $A$ temporally precedes $Y$, which implies $Y \notin \An_{\g}(A)$,
since otherwise $A$
cannot have any causal effect on $Y$ and the problem is trivial.
Recall that $Z$ consists of observable and unobservable pre-treatment covariates
such that $S \subseteq Z$ and $Z \cup \{A, Y\}$ is causally closed.
Therefore, $\g$ contains not only the observable variables $\{A,Y\}
\cup S$ but also the unobservable variables $Z \setminus S$.

Next, we give a review of a number of confounder selection methods,
which are organized according to their assumptions on the amount of
structural knowledge of $\mathcal{G}$.

\subsection{Methods that assume full structural knowledge}
\label{sec:full-structural}

With full structural knowledge, achieving the primary objective of
confounding control is a solved problem by the well-known back-door
criterion due to \citet{pearl1993bayesian,pearl2009}:

\begin{proposition}[Back-door criterion] \label{prop:backdoor}
A set $C$ of pre-treatment covariates controls for confounding if $C$
blocks all the back-door paths from $A$ to $Y$ in $\g$. Furthermore,
the reverse is true under the faithfulness
assumption.\footnote{Here the faithfulness is with respect to the corresponding SWIG; see \cref{apx:swig}.}
\end{proposition}

Here, a path from $A$ to $Y$ in $\g$ is called a \emph{back-door path}
if it starts with an edge into $A$.
Note that given a sufficient adjustment set $C$,
a superset of $C$ need not continue to block all the back-door paths. This is known
as the ``collider bias'' or ``M-bias'' \citep{greenland1999causal} and an example
is given in \cref{fig:sel}(a).

For the generalization that allows adjusting
for post-treatment variables, see \citet{shpitser2010validity}; for
the extension to graphs that describe Markov equivalence classes, see also
\citet{perkovic2018complete}.

When full structural knowledge is available, the graph can also
inform secondary objectives of confounder selection. In particular,
there is a line of work that selects the adjustment set to minimize
the asymptotic variance in estimating the causal effect
\citep{kuroki2003covariate,henckel2022,witte2020efficient,rotnitzky2020efficient}.
When all pre-treatment covariates are observed (i.e.\ $Z=S$), there is
an optimal choice of $C$ that minimizes the variance of the causal
effect estimator under linear and semiparametric models. This so-called
\emph{optimal adjustment set} consists of parents of the
outcome $Y$ and parents of all mediators between $A$ and $Y$, excluding $A$ and any descendant of a mediator; see \citet[\S3.4]{henckel2022}.

\subsection{Methods that assume partial structural knowledge}
\label{sec:partial-structural}

Of course, the full graph $\g$ is rarely known in practice. With just
partial knowledge of the causal graph, there exist several methods
for confounder selection. It is worthwhile to scrutinize the
assumptions that guarantee the validity of these methods.

\begin{assumption} \label{assump:S=Z}
  All relevant pre-treatment covariates are (or can be) measured, i.e., $S=Z$.
\end{assumption}
\begin{assumption} \label{assump:S-already}
  $S$ is a sufficient adjustment set.
\end{assumption}
\begin{assumption} \label{assump:some-subset-S}
  At least one subset of $S$ is a sufficient adjustment set.
\end{assumption}

The assumptions are listed from the strongest to the weakest:
\Cref{assump:S=Z} implies \Cref{assump:S-already} (see
\cref{prop:Z-adj} in the appendix), and \Cref{assump:S-already} trivially implies \Cref{assump:some-subset-S}. \Cref{assump:some-subset-S} is the weakest possible
in the sense that without this assumption, it is impossible to achieve
the primary objective of confounding control.

In graphical terms, the \emph{conjunctive} (or \emph{common}) \emph{cause criterion} selects
\begin{equation*} \label{eqs:C-comm}
  \Ccommon \equiv  S \cap \An_{\g}(A) \cap \An_{\g}(Y),
\end{equation*}
and the \emph{pre-treatment criterion} selects
\begin{equation*} \label{eqs:C-pre}
  \Cpre \equiv S.
\end{equation*}
Although the pre-treatment criterion does not use any structural
knowledge other than the temporal order, it is included here to
facilitate comparison with other criteria.

\begin{proposition} \label{prop:common-cause}
  Under \Cref{assump:S=Z}, $\Ccommon$ is a sufficient adjustment set.
\end{proposition}

\begin{proposition} \label{prop:pretreatment}
  Under \Cref{assump:S-already}, $\Cpre$ is a sufficient adjustment
  set.
\end{proposition}

\begin{figure}[t]
  \centering
  \begin{tikzpicture}
    \tikzset{rv/.style={circle,inner sep=1pt,fill=gray!20,draw,font=\sffamily},
      fv/.style={circle,inner sep=1pt,fill=gray!40,draw,font=\sffamily},
      node distance=12mm, >=stealth}
    \begin{scope}
      \node[rv] (A) {$A$};
      \node[right of=A] (H) {};
      \node[rv, right of=H] (Y) {$Y$};
      \node[fv, above of=A, yshift=5mm] (U1) {$U_1$};
      \node[fv, above of=Y, yshift=5mm] (U2) {$U_2$};
      \node[rv, above of=H] (L) {$L$};
      \draw[->,very thick] (U1) -- (L);
      \draw[->,very thick] (U2) -- (L);
      \draw[->,very thick] (U1) -- (A);
      \draw[->,very thick] (U2) -- (Y);
      \draw[->,very thick] (A) -- (Y);
      \node[below=2mm of H] {(a)};
    \end{scope} \begin{scope}[xshift=5cm]
      \node[rv] (A) {$A$};
      \node[right of=A] (H) {};
      \node[rv, right of=H] (Y) {$Y$};
      \node[rv, above of=A, yshift=5mm] (X1) {$X_1$};
      \node[rv, above of=Y, yshift=5mm] (X2) {$X_2$};
      \draw[->,very thick] (X1) -- (A);
      \draw[->,very thick] (X2) -- (Y);
      \draw[->,very thick] (A) -- (Y);
      \draw[->,very thick] (X1) -- (X2);
      \node[below=2mm of H] {(b)};
    \end{scope} \end{tikzpicture}
  \caption{(a) M-bias: $U_1,U_2$ are unobserved and $S = \{L\}$. It holds that $A \indep Y(a)$ for $a=0,1$ unconditionally. However, controlling for $L$ opens the back-door path $A \leftarrow U_1 \rightarrow L \leftarrow U_2 \rightarrow Y$ and hence would introduce bias. (b) Suppose $S = \{X_1, X_2\}$ and we have $\PredA(S) = \{X_1\}$, $\PredYA(S) = \{X_2\}$. Using the definitions from Section~\ref{sec:no-structural}, we note that the conjunctive $\Cconj(S) = \PredA(S) \cap \PredYA(S) = \emptyset$ fails to control for confounding. Meanwhile, $\CAY(S) = \{X_1\}$, $\CYA(S) = \{X_2\}$ and $\Cdisj(S) = \{X_1, X_2\}$ all suffice to control for confounding.}
  \label{fig:sel}
\end{figure}
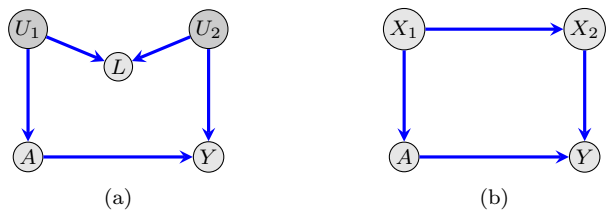

\Cref{prop:pretreatment} is just a tautology.
\Cref{prop:common-cause} is proved in \Cref{apx:graphical}.
The assumptions stated in these two propositions are the weakest possible
on our list. Indeed, the conjunctive cause can fail under \cref{assump:S-already}:
consider $S = \{X_2\}$ in \cref{fig:sel}(b), which leads to $\Ccommon = \emptyset$ and fails to block the back-door path.
Also, the pre-treatment criterion may not hold under \cref{assump:some-subset-S}:
for $S = \{L\}$ in \cref{fig:sel}(a), while $\emptyset \subset S$ is a sufficient adjustment set, $S$ is not.

Alternatively, one may use the \emph{disjunctive cause criterion}:
\begin{equation*}
  \CvdW \equiv S \cap \left[\An_{\g}(A) \cup \An_{\g}(Y) \right].
\end{equation*}
\citet{vanderweele2011new} proved the following remarkable result
regarding this criterion.

\begin{proposition} \label{prop:disj-cause}
  Under \cref{assump:some-subset-S} and the faithfulness assumption,
  $\CvdW$ is a sufficient adjustment set.
\end{proposition}

As pointed out by \citet{richardson2018discussion}, this result
follows almost immediately from a general result on inducing
paths; for completeness, a proof is provided in
\cref{apx:disj-cause}.

In view of the last three propositions, the disjunctive cause
criterion is superior to the
pre-treatment criterion: as a subset of $\Cpre$, $\CvdW$ is
guaranteed to control for confounding under an even weaker
assumption. Both the conjunctive and disjunctive
cause criteria require substantive knowledge about the causes of $A$
and $Y$, though much less demanding than knowing the full
graph. Between the two, the disjunctive cause criterion is usually
preferred because it achieves the primary objective
under a weaker assumption. That being said, when \cref{assump:some-subset-S} fails to hold,
using $\CvdW$ may lead to a larger bias than $\Cpre$; see \citet[\S2.1]{richardson2018discussion}.

While the aforementioned criteria
only require partial structural knowledge, verifying the corresponding
assumptions requires much more information and can be very
difficult. Further, these criteria are formulated with a
pre-specified set of observed variables $S$. Hence they are not
suitable for the design of observational studies where one must choose
a set of covariates for measurement. This is an important limitation
because design is often more consequential than analysis in
observational studies \citep{rubin2008objective}.

To address these issues, \citet{guo2023confounder} recently introduced
a new procedure to select confounders that does not require
pre-specifying the graph or the set $S$. It is based on a
reformulation of the back-door criterion (\cref{prop:backdoor}) under
latent projection. In fact, this procedure inverts the process of
latent projection by iteratively expanding the causal graph as more
and more structural knowledge is elicited from the user. This is
repeated until a sufficient adjustment set is found or it can be
determined that no sufficient adjustment set exists. Compared to the
``static'' criteria $\CvdW$, $ \Ccommon$ and $\Cpre$, this new
procedure is able to elicit structural knowledge in an interactive and
economical way.

\subsection{Methods that assume no structural knowledge}
\label{sec:no-structural}

Without any knowledge about the causal structure, it must be assumed that
the primary objective is \emph{already satisfied} by a given set of
covariates because the ignorability condition \eqref{eq:no-confounding}
cannot be tested using observational data alone. Without loss of generality,
we assume $S$ is a sufficient adjustment set; that is, for this subsection,
we shall proceed under \cref{assump:S-already}.

In this case, data-driven approaches can be useful to achieve some of
the secondary objectives, especially if one would like to reduce the
dimension of $C$. This is typically accomplished by considering two
types of conditional independencies.

\begin{definition}[Treatment/outcome Markov blankets and boundary] \label{def:markov-blanket}
For any $V \subseteq S$, the collection of \emph{treatment Markov blankets} in $V$ is defined as
\begin{equation*}
  \predA(V) \equiv \{V' \subseteq V: A \indep V \setminus V' \mid V'
  \}
\end{equation*}
and the collection of \emph{outcome Markov blankets} in $V$ is defined as
\begin{equation*}
  \predYA(V) \equiv \{V' \subseteq V: Y \indep V \setminus V' \mid A, V'\}.
\end{equation*}
Further, we define the \emph{treatment Markov boundary}
\begin{equation} \label{eqs:PredA}
\PredA(V) \equiv \bigcap_{V' \in \predA(V)} V'
\end{equation}
and the \emph{outcome Markov boundary}
\begin{equation} \label{eqs:PredYA}
\PredYA(V) \equiv \bigcap_{V' \in \predYA(V)} V'.
\end{equation}
\end{definition}

As a convention, conditional independence always holds with an empty
vector, so $V \in \predA(V)$ and $V \in \predYA(V)$.
For the Markov blankets to have desirable properties, we require the positivity of the distribution over $(A,S)$ in the following sense, which warrants the application of the intersection property of conditional independence; see also \citet[\S2.3.5]{studeny2004book}. In below, let $P_A$ be the marginal distribution of $A$ and $P_v$ be the marginal distribution of $v$ for each $v \in S$.

\begin{assumption}[Positivity] \label{assump:positivity}
The distribution over $(A,S)$ admits a density that is almost
everywhere positive with respect to the product measure $P_A \prod_{v
  \in S} P_{v}$.
\end{assumption}

\begin{lemma} \label{lem:close-intersect}
Under \cref{assump:positivity}, $\predA(V)$ and $\predYA(V)$ are closed under intersection for every $V \subseteq S$.
\end{lemma}

Hence, under positivity $\PredA(V)$ and $\PredYA(V)$ are the
\emph{minimum} of $\predA(V)$ and $\predYA(V)$,
respectively.
Moreover, the following result suggests that the order of reduction
within $\predA(V)$ and $\predYA(V)$ is irrelevant.

\begin{lemma}[Reduction] \label{lem:reduce}
Under \cref{assump:positivity}, we have
\[\PredA(V) = \PredA(V'), \quad \PredYA(V) = \PredYA(V'') \]
for every $V' \in \predA(V)$, $V'' \in \predYA(V)$ and $V \subseteq S$.
\end{lemma}

Given that $S$ is a sufficient adjustment set, the following two key lemmas state that every treatment or outcome Markov blanket in $S$ is also a sufficient
adjustment set. The proof of these lemmas applies graphoid properties  of
conditional independence listed in \cref{apx:graphoid} and is deferred
to \cref{apx:blanket}.

\begin{lemma}[Soundness of $\predA$] \label{lem:predA}
Under \cref{assump:S-already}, every set in $\predA(S)$ controls for confounding.
\end{lemma}

\begin{lemma}[Soundness of $\predYA$] \label{lem:predYA}
Under \cref{assump:S-already,assump:positivity}, every set in $\mathfrak{R}_{Y|A}(S)$ controls for confounding.
\end{lemma}

\begin{remark}[Binary treatment]
That $A$ is binary is necessary for \cref{lem:predYA} to hold. When $A$ has more than two levels, under the same assumptions, for every $S_{Y|A} \in \predYA(S)$ one can still show
\begin{equation} \label{eqs:weak-control}
Y_a \indep \{A=a\} \mid S_{Y|A}
\end{equation}
for every level $a$. This is weaker than $Y_a \indep A \mid
S_{Y|A}$. Nevertheless, equation \eqref{eqs:weak-control} still
ensures the identification of the distribution of $Y_a$:
\begin{align*}
  \mathbb{P} (Y_a = y) &= \E[\mathbb{P}(Y_a = y \mid S_{Y|A})]\\
                       &= \E[\mathbb{P}(Y_a = y \mid A=a, S_{Y|A})] \\
                       &= \E[\mathbb{P}(Y = y \mid A=a, S_{Y|A})].
\end{align*}
However, \eqref{eqs:weak-control} is not sufficient to identify
$\E[Y_a \mid A = a']$ for $a \neq a'$ (required for identifying the
treatment effect on the treated) when $A$ has more than two levels.
\end{remark}

We can use $\predA$ and $\predYA$ to find subsets of $S$ that control for confounding.
There are at least four possible ways to do this:
\begin{description}
\item[Conjunctive] $\Cconj(S) \equiv \PredA(S) \cap \PredYA(S)$,
\item[Disjunctive] $\Cdisj(S) \equiv \PredA(S) \cup \PredYA(S)$,
\item[$\AY$-sequential] $\CAY(S) \equiv \PredYA(\PredA(S))$,
\item[$\YA$-sequential] $\CYA(S) \equiv \PredA(\PredYA(S))$.
\end{description}
More generally, the Markov boundaries $\PredA$ and $\PredYA$ above can be replaced by any set from the corresponding collection of Markov blankets $\predA$ and $\predYA$. For example, the $\YA$-sequential rule can take any set from $\predA(S_{Y|A})$ for any $S_{Y|A} \in \predYA(S)$.

\begin{example} \label{ex:four-sel-M}
Consider \Cref{fig:sel}(a) with $S = \{U_1,U_2,L\}$. We have
$\PredA(S) = \{U_1\}$, $\PredYA(S) = \{U_2\}$ and hence
\[ \Cconj(S) = \emptyset, \quad \Cdisj(S) = \{U_1, U_2\}. \]
Further, we have $\CAY(S) = \PredYA(\{U_1\}) = \emptyset$ and $\CYA(S) = \PredA(\{U_2\}) = \emptyset$.
\end{example}

\begin{example} \label{ex:four-sel}
Consider the case of \Cref{fig:sel}(b) with $S = \{X_1, X_2\}$. By d-separation,
$A \indep X_2 \mid X_1$ and $Y \indep X_1 \mid A, X_2$, so
\begin{align*}
  \predA(S) &= \{\{X_1\}, \{X_1,X_2\}\}, \quad \PredA(S) = \{X_1\},\\
  \predYA(S) &= \{\{X_2\}, \{X_1,X_2\}\}, \quad \PredYA(S) = \{X_2\}.
\end{align*}
It follows that $\Cconj(S) = \emptyset$, $\Cdisj(S) = \{X_1, X_2\}$, $\CAY(S) = \{X_1\}$ and $\CYA(S) = \{X_2\}$, among which only $\Cconj(S)$ fails to control for confounding.
\end{example}

Among these four methods, it is clear from \cref{ex:four-sel} that the conjunctive criterion $\Cconj(S)$ can fail to control for confounding, even when $S \cup \{A,Y\}$ is causally closed. The
validity of the other three criteria directly follows from
\Cref{lem:predA,lem:predYA} and applying the weak union property of conditional independence.

\begin{theorem} \label{thm:four-sel}
  Under \cref{assump:S-already,assump:positivity}, $\CAY(S)$, $\CYA(S)$ and $\Cdisj(S)$ are sufficient adjustment sets.
\end{theorem}

For the rest of this section, suppose \cref{assump:positivity} (positivity)
holds. By applying the soundness of $\predA$ and $\predYA$
(\Cref{lem:predA,lem:predYA}) sequentially, one can show that
$\PredA(\CAY)$, $\PredYA(\PredA(\CAY))$ and so forth are sufficient
adjustment sets too.
Furthermore, consider alternating between $\PredA$ and $\PredYA$
and let $\CAY^{\ast}(S)$ and $\CYA^{\ast}(S)$ be, respectively, the
limit of
\[ S \xrightarrow{\PredA} C_{\text{A}} \xrightarrow{\PredYA} C_{\text{AY}} \xrightarrow{\PredA} C_{\text{AYA}} \xrightarrow{\PredYA} \dots \]
and
\[ S \xrightarrow{\PredYA} C_{\text{Y}} \xrightarrow{\PredA}
  C_{\text{YA}} \xrightarrow{\PredYA} C_{\text{YAY}}
  \xrightarrow{\PredA} \dots. \]
The iterations terminate in a finite number of steps, when $|\predA|=1$ or $|\predYA|=1$.
Then $\CAY^{\ast}(S)$ and $\CYA^{\ast}(S)$ are also sufficient
adjustment sets and cannot be further reduced using these operations.

\begin{lemma}[Stability] \label{lem:stability}
Under \cref{assump:positivity}, for $C^{\ast}\in \{\CAY^{\ast}(S), \CYA^{\ast}(S)\}$, we have $\PredA(C^{\ast}) = \PredYA(C^{\ast}) = C^{\ast}$.
\end{lemma}

\begin{example}[continues=ex:four-sel-M]
$\CAY^{\ast}(S) = \CYA^{\ast}(S) = \emptyset$ in  \Cref{fig:sel}(a).
\end{example}

\begin{example}[continues=ex:four-sel]
$\CAY^{\ast}(S) = \{X_1\}$, $\CYA^{\ast}(S) = \{X_2\}$ in  \Cref{fig:sel}(b).
\end{example}

In fact, from a graphical perspective, $\CAY^{\ast}(S)$ and $\CYA^{\ast}(S)$ are \emph{minimal}
sufficient adjustment sets in the sense that no proper subset can
control for confounding, or in view of \cref{prop:backdoor}, can block
all the back-door paths between $A$ and $Y$.

\begin{theorem}[Minimal sufficient adjustment sets] \label{thm:minimal}
Suppose \cref{assump:S-already,assump:positivity} hold. The set $C^{\ast}
\in \{\CAY^{\ast}(S)$, $ \CYA^{\ast}(S)\}$ controls for confounding.

Further, let $\g$ be a causal DAG over
variables $Z \cup \{A,Y\}$. For $a=0,1$, suppose
the distribution of $(A, Y_a, Z)$ is Markov to and faithful to the
SWIG $\g(a)$. Then, for $C^{\ast} \in \{\CAY^{\ast}(S),
\CYA^{\ast}(S)\}$ there is no proper subset of $C^{\ast}$ that also
controls for confounding.
 \end{theorem}

The soundness of $\CAY^{\ast}(S)$ and $\CYA^{\ast}(S)$ does
not rely on the assumption of a causal graph; nor is structural
knowledge of the graph required for computing $\CAY^{\ast}(S)$ and
$\CYA^{\ast}(S)$. In relation to the results above, see also
\citet[\S4.2]{de2011covariate}, who used weak transitivity to show
that no single element can be removed from $\CAY^{\ast}(S)$ or $\CYA^{\ast}(S)$ such that the remaining set continues to control for confounding. Yet, that result is weaker than \cref{thm:minimal} as it does not rule out an even smaller subset controlling for confounding. To show  \cref{thm:minimal}, we prove a strengthened version of weak transitivity; see \cref{apx:minimal}.

To compute the Markov boundaries, a principled approach can
be developed using the following result.

\begin{proposition} \label{prop:coordinate-wise}
  Under \cref{assump:positivity}, for $V \subseteq S$, it holds that
  \begin{align*}
    \PredA(V) &= \{v \in V: A \not \indep v \mid V \setminus \{v\} \}, \\
    \PredYA(V) &= \{v \in V: Y \not \indep v \mid A, \, V \setminus \{v\} \}.
  \end{align*}
\end{proposition}

See \cref{apx:blanket} for its proof. As an alternative to \Cref{prop:coordinate-wise}, by \Cref{lem:reduce},
we can also perform the conditional
independence tests sequentially and only use the already reduced set
as $V$. This leads to a backward stepwise procedure in \Cref{alg:Pred}
to compute $\PredA(V)$ and $\PredYA(V)$; see also
\citet[\S4]{vanderweele2011new}.

\begin{algorithm}[!htb] \DontPrintSemicolon \SetNoFillComment \caption{Backward stepwise algorithm for finding the treatment/outcome Markov boundary. } \label{alg:Pred} \KwIn{$k$ (0 for $\PredA$, 1 for $\PredYA$), a list of variables $V$ }
\KwOut{$\PredA(V)$ ($k=0$) or $\PredYA(V)$ ($k=1$)}
$R \gets \{\}$ \;
\While{$V$ is not empty}{
$v \gets \text{the first element of } V$ \;
Remove $v$ from $V$ \;
        \If{$(k=0$ and $A \not\indep v \mid V, R)$ or $(k=1$ and $Y \not\indep v \mid A, V, R)$} {
                $R \gets R \cup \{v\}$ \;
        }
}
\Return{$R$}\;
\end{algorithm}

For practical use, one must employ a parametric or nonparametric test for conditional independence.
As a caveat, if the statistical test used has low power to detect conditional dependence (i.e., claiming false conditional independencies),
the Markov boundaries selected by the aforementioned methods will be too small and not
sufficient to control for confounding. The same issue underlies many other constraint-based structure learning methods as well \citep{drton2017structure,strobl2019estimating}.
For estimating the causal effect with a data-driven set of confounders, the quality of effect estimation depends on the estimator, the set $S$, the statistical procedure for selecting $S$, as well as the underlying data generating mechanism. The reader is referred to \citet{witte2019covariate} for simulation studies on these aspects.

\section{Discussion}
\label{sec:discussion}

Confounder selection is a fundamental problem in causal
inference, and the backdoor criterion solves this problem when
full structural knowledge is available. In practice, however, we
almost always have only partial knowledge of the causal
structure. This is when different
principles and objectives of confounder selection clash the
most. More theoretical and
methodological development is needed when only partial
knowledge of the causal structure is available.

Compared to the widely used conjunctive cause and pre-treatment
criteria, the disjunctive cause criterion achieves the primary
objective of confounding control under the weakest possible
assumption. However, as argued by
\citet[\S2.1]{richardson2018discussion}, the appeal of the disjunctive
cause criterion diminishes when we consider the secondary objectives
that arise in practical situations.

In view of the propositions in \Cref{sec:partial-structural}, the
conjunctive/common cause criterion appears to require the strongest
assumption when a set $S$ of observed or potentially observable
covariates is given. Nor is taking the conjunction useful in the context of
data-driven variable selection (\Cref{ex:four-sel}). Thus, there is
little theoretical support for the conclusions in \citet{austin2011}
as described in the introduction. However, as
mentioned in \cref{sec:partial-structural}, this may not be how
confounder selection is or should be conceptualized in the design of
observational studies. In that process, investigators try to
directly come up with potential confounders, usually through arguing
about common causes, before reviewing whether or how those confounders
can be measured.

Although the conjunctive/common cause heuristic has little value in
\emph{analyzing} observational studies, we believe that they are still
indispensible in \emph{designing} observational studies. This can be
seen from the iterative graph expansion procedure in
\citet{guo2023confounder}, where the common cause heuristic is
formalized as identifying the so-called \emph{primary adjustment sets}
that block all immediately confounding paths between two
variables. However, \citet{guo2023confounder} do not take this as a
one-shot decision; instead, they show that by iterating the common
cause heuristic one can obtain all the minimal sufficient adjustment sets.

Although data-driven confounder selection can be used without any
structural knowledge, it hinges on the premise that the set of variables under consideration is
already a sufficient adjustment set. In any case, substantive
knowledge is essential in observational studies as conditional
ignorability cannot be empirically verified with observational
data alone. Moreover, in general, substantive knowledge cannot be replaced by structure
learning methods that learn the causal graph before selecting
confounders
\citep{maathuis2015generalized,nandy2017estimating,perkovic2018complete},
because (1) the underlying
causal graph is typically not identified from data due to Markov
equivalence \citep{frydenberg1990chain,verma1990equivalence}; (2)
causal structure learning methods themselves rely on strong
assumptions about the causal structure (such as no hidden variables or faithfulness).

\citet{greenland07_invit_commen} and \citet{greenland2016outcome}
challenged the usefulness of data-driven confounder selection (without
structural knowledge) and suggested that other modern techniques such
as regularization and model averaging are more suitable than variable
selection. The same philosophy underlies influence-function-based
causal effect estimators such as double machine learning
\citep{chernozhukov18_doubl_machin_learn_treat_struc_param},
one-step corrected estimation \citep{kennedy2022semiparametric} and targeted maximum likelihood estimation \citep{laan11_target_learn}. In view
of the secondary objectives in \Cref{sec:object-conf-select}, this
criticism is reasonable if we assume the full set $S$ already controls
for confounding, because statistical confounder selection simplifies the
final estimator (especially if one uses matching-based methods) at the cost of
sacrificing robustness or stability. That being said, statistical variable
selection could still be useful in other scenarios by complementing partial structural knowledge.

\appendix
\section{Conditional independence and Markov blanket} \label{apx:subsets}
In the following, when a set operation is applied to a random vector, it means the random vector consisting of entries from the set operation applied to the coordinates.

\subsection{Graphoid properties of conditional independence} \label{apx:graphoid}
Let $W,X,Y,Z$ be random variables. Conditional independence satisfies the following graphoid properties:
\begin{description}
\item[Symmetry] $X \indep Y \mid Z \implies Y \indep X \mid Z$.
\item[Decomposition] $X \indep Y, W \mid Z \implies X \indep Y \mid Z \text{ and } X \indep W \mid Z$.
\item[Weak union] $X \indep Y, W \mid Z \implies X \indep Y \mid W, Z$.
\item[Contraction] $X \indep Y \mid Z \text{ and } X \indep W \mid Y, Z \implies X \indep W, Y \mid Z$.
\item[Intersection] $X \indep Y \mid W, Z \text{ and } X \indep W \mid Y, Z \implies X \indep W, Y \mid Z$.
\end{description}
Among the above, the intersection property additionally requires the positivity of $P_{W,Y,Z}$ in the sense that $P_{W,Y,Z}$ admits a density with respect to $P_{W} P_{Y} P_{Z}$ that is almost everywhere positive; see also \citet[\S2.3.5]{studeny2004book}.
For properties of conditional independence, see also \citet[\S1.1.5]{pearl2009}.

\subsection{Markov blanket and boundary} \label{apx:blanket}
\begin{proof}[Proof of \cref{lem:reduce}]
We prove $\PredA(V) = \PredA(V')$ for every $V' \in \predA(V)$; the proof for $\PredYA$ follows similarly. We first show $\PredA(V) \subseteq \PredA(V')$. By definition of $\PredA$, it suffices to show $\predA(V) \supseteq \predA(V')$. Take $\tilde{V} \in \predA(V')$ so we have
\[ A \indep V' \setminus \tilde{V} \mid \tilde{V}. \]
Since $V' \in \predA(V)$, we also have
\[ A \indep V \setminus V' \mid V'. \]
It follows from contraction that
\[ A \indep V \setminus \tilde{V} \mid \tilde{V} \]
and hence $\tilde{V} \in \predA(V)$. Now we show $\PredA(V) \supseteq \PredA(V')$. This is true because $\PredA(V) \in \predA(V')$, which follows by weak union.
\end{proof}

\begin{proof}[Proof of \cref{lem:predA}]
  Fix $S_A \in \predA(S)$ and $a \in \{0,1\}$. By contraction, $A \indep Y_a \mid S$ and $A \indep S \setminus S_A \mid S_A$ together imply $A \indep Y_a,\, S \setminus S_A \mid S_A$, which further gives $A \indep Y_a \mid S_A$ by decomposition.
\end{proof}

\begin{proof}[Proof of \cref{lem:predYA}]
  Fix $S_{Y|A} \in \predYA(S)$ and $a \in \{0,1\}$. That $S$ controls for confounding implies
  \[ Y_a \indep \{A=a\} \mid S \setminus S_{Y|A},\, S_{Y|A}. \]
  By definition of $\predYA$, we have $Y \indep S \setminus S_{Y|A} \mid A=a,\, S_{Y|A}$,
  which by consistency becomes
  \[ Y_a \indep S \setminus S_{Y|A} \mid A=a,\, S_{Y|A}. \]
  Under positivity, by applying the intersection property to the previous two displays, we have
  \[ Y_a \indep \{A=a\},\, S \setminus S_{Y|A} \mid S_{Y|A}, \]
  which by decomposition further implies $Y_a \indep \{A=a\} \mid
  S_{Y|A}$, which is equivalent to $Y_a \indep A \mid S_{Y|A}$ since
  $A$ is binary.
\end{proof}

\begin{proof}[Proof of \cref{thm:four-sel}]
 To see $\Cdisj(S)$ controls for confounding, note $\PredA(S) \cup \PredYA(S) \in \predA(S)$ by weak union and then apply \cref{lem:predA}.
 That $\CAY(S)$ and $\CYA(S)$ control for confounding under positivity follows from sequentially applying \cref{lem:predA,lem:predYA}.
\end{proof}

\begin{proof}[Proof of \cref{lem:stability}]
Markov boundaries $\PredA(\cdot)$ and $\PredYA(\cdot)$ are well-defined under \cref{assump:positivity}. By the termination condition, either $\PredA(C^{\ast}) = C^{\ast}$ or $\PredYA(C^{\ast}) = C^{\ast}$ holds already. First, suppose $\PredA(C^{\ast}) = C^{\ast}$ holds. We prove $\PredYA(C^{\ast}) = C^{\ast}$ by contradiction. Since $\PredYA(C^{\ast}) \subseteq C^{\ast}$, suppose $\PredYA(C^{\ast}) = C' \subsetneq C^{\ast}$. By the last iteration, $C^{\ast} = \PredYA(\tilde{C})$ for some $\tilde{C} \supsetneq C^{\ast}$. Writing $C^{\ast} = C' \cup (C^{\ast} \setminus C')$, we have
\[ Y \indep C^{\ast} \setminus C' \mid A,\, C' \]
and
\[ Y \indep \tilde{C} \setminus C^{\ast} \mid A,\, C',\, C^{\ast} \setminus C'. \]
Applying contraction, we have
\[ Y \indep \tilde{C} \setminus C' \mid A,\, C', \]
which contradicts $C^{\ast} = \PredYA(\tilde{C})$ and \cref{eqs:PredYA} because $C'$ is a proper subset of $C^{\ast}$. The other case follows similarly.
\end{proof}

\begin{proof}[Proof of \cref{lem:close-intersect}]
It directly follows from the next lemma and the fact that the positivity of $(A,S)$ implies the positivity of $S$.
\end{proof}

\begin{lemma} \label{lem:sep-set-inter}
Let $X,Z$ be two random variables and let $S$ be a finite-dimensional random vector. Define
  \[ \mathfrak{R}_{X|Z}(S) \equiv \{S' \subseteq S: X \indep S \setminus S' \mid S', \, Z \}. \]
Suppose the distribution of $(Z,S)$ is positive in the sense that $P_{Z,S}$ admits a density with respect to $P_Z \prod_{v \in S} P_v$ that is almost everywhere positive.
It holds that $\mathfrak{R}_{X|Z}(S)$ is closed under intersection.
\end{lemma}
\begin{proof}
  Take $S_1, S_2 \in \mathfrak{R}_{X|Z}(S)$. By definition,
  \begin{align}
    X \indep S \setminus S_1 \mid S_1,\, Z \label{eqs:app-1} \\
    X \indep S \setminus S_2 \mid S_2,\, Z \label{eqs:app-2}
  \end{align}
  which can be rewritten as
  \begin{align*}
    X \indep S_2 \setminus S_1,\, S \setminus (S_1 \cup S_2) \mid S_1 \cap S_2,\, S_1 \setminus S_2,\, Z, \\
    X \indep S_1 \setminus S_2,\, S \setminus (S_1 \cup S_2) \mid S_1 \cap S_2,\, S_2 \setminus S_1,\, Z.
  \end{align*}
  Further by decomposition,
  \begin{align*}
    X \indep S_2 \setminus S_1 \mid S_1 \cap S_2,\, S_1 \setminus S_2,\, Z\\
    X \indep S_1 \setminus S_2 \mid S_1 \cap S_2,\, S_2 \setminus S_1,\, Z.
  \end{align*}
  Applying intersection under positivity, we have
  \[ X \indep S_1 \setminus S_2, \, S_2 \setminus S_1 \mid S_1 \cap S_2, \, Z. \]
  which by decomposition implies
  \[ X \indep S_1 \setminus S_2 \mid S_1 \cap S_2, \, Z. \]
  Meanwhile, note that \eqref{eqs:app-1} can be rewritten as
  \[ X \indep S \setminus S_1 \mid S_1 \setminus S_2, \, S_1 \cap S_2, \, Z.\]
  Applying contraction to the two previous displays, we get
  \[ X \indep S_1 \setminus S_2, \, S \setminus S_1 \mid S_2 \cap S_2, \, Z,\]
  i.e.,
  \[ X \indep S \setminus (S_1 \cap S_2) \mid S_1 \cap S_2, \, Z. \]
  Hence, $S_1 \cap S_2 \in \mathfrak{R}_{X|Z}(S)$.
\end{proof}

\begin{proof}[Proof of \cref{prop:coordinate-wise}]
  We show $\PredA(V) = \Gamma_A(V)$, where $\Gamma_A(V) \equiv \{v \in V: A \not \indep v \mid V \setminus \{v\} \}$. First, we show $\Gamma_A(V) \subseteq \PredA(V)$. Fix $v \in \Gamma_A(V)$. Suppose $v \notin \PredA(V)$. We can then rewrite $A \indep V \setminus \PredA(V) \mid \PredA(V)$ as
  \[ A \indep v, \, V \setminus (\PredA(V) \cup \{v\}) \mid \PredA(V). \]
  By weak union, it follows that $A \indep v \mid V \setminus \{v\}$, contradicting $v \in \Gamma_A(V)$. Now we show $\PredA(V) \subseteq \Gamma_A(V)$. Fix $v \in \PredA(V)$. By definition of $\PredA(V)$, we have
  \[ A \indep V \setminus \PredA(V) \mid \PredA(V) \setminus \{v\},\, v.\]
  Suppose $v \notin \Gamma_A(V)$, i.e., $A \indep v \mid V \setminus \{v\}$, which can be rewritten as
  \[ A \indep v \mid \PredA(V) \setminus \{v\}, \, V \setminus \PredA(V). \]
  Note that \cref{assump:positivity} implies the distribution over $A$ and $V$ is also positive. Applying intersection to the previous two displays, we get
  \[ A \indep v,\, \PredA(V) \mid \PredA(V) \setminus \{v\}, \]
  which contradicts the minimality of $\PredA(V)$ in $\predA(V)$ in \cref{eqs:PredA}.

  By $P(A=a,Y,V) = P(A=a) P(Y,V \mid A=a) = P(A=a) P(Y_a,V \mid A=a)$ for $a=0,1$, \cref{assump:positivity} implies that the distribution over $A,Y,V$ is also positive. The proof for the equality on $\PredYA(V)$ follows similarly.
\end{proof}

\section{Graphical results} \label{apx:graphical}
\subsection{DAG and d-separation} \label{apx:d-sep}
We use standard graphical terminology; see also, e.g.,
\citet[\S1.2]{pearl2009}. In particular, for two vertices $u$ and $v$,
we say $u$ is an ancestor of $v$, or equivalently $v$ is a descendant
of $u$, if either $u = v$ or there is a causal path $u \rightarrow
\dots \rightarrow v$. The set of ancestors of vertex $u$ in graph $\g$
is denoted as $\An_{\g}(u)$. Similarly, the set of descendants is
denoted as $\De_{\g}(u)$. By definition, $u \in \An_{\g}(u)$ and $u
\in \De_{\g}(u)$. We use symbol $\nd_{\g}(u)$ for non-descendants of
$u$, i.e., the complement of $\De_{\g}(u)$. The definitions of relational
sets extend disjunctively to a set of vertices, e.g.,
\[ \An_{\g}(L) \equiv \cup_{v \in L} \An_{\g}(v), \quad  \De_{\g}(L) \equiv \cup_{v \in L} \De_{\g}(v).\]

For two vertices $u$ and $v$, a path $\pi$ between them consists of a sequence
of distinct vertices such that consecutive vertices are adjacent in the graph.
For vertices $w,z$ also on path $\pi$, we use notation  $\pi(w,z)$ for the subpath between $w$ and $z$.
A non-endpoint vertex $k$ is called a collider on the path if it is of the form
$ \cdots \circ \rightarrow k \leftarrow \circ \cdots$; otherwise $k$ is called a
non-collider.
\begin{definition}[d-connecting path] \label{def:d-conn}
A path $p$ between $u$ and $v$ is called d-connecting given a set of vertices $L$ ($u, v \notin L$) if
(1) every non-collider on $p$ is excluded from $L$ and (2) every collider on $p$ is in $L$ or is an ancestor of some vertex in $L$.
\end{definition}
\noindent When a path $p$ is not d-connecting given $L$, we also say $L$ blocks the path $p$.

\begin{definition}[d-separation] \label{def:d-sep}
If there is no d-connecting path between $u$ and $v$ given $L$ ($u, v \notin L$), we say that $u$ and $v$ are d-separated given $L$. Similarly, for disjoint vertex sets $U,V,L$, we say that $U$ and $V$ are d-separated given $L$, if there is no d-connecting between $u$ and $v$ given $L$ for $u \in U$, $v \in V$.
\end{definition}
\noindent We use symbol $u \indep_{\g} v \mid L$ to denote that $u$ and $v$ are d-separated by $L$ in graph $\g$. A similar notation applies to the d-separation between sets. It holds that d-separation shares the graphoid properties of conditional independence listed in \cref{apx:graphoid}, where symbol $\indep$ is replaced with $\indep_{\g}$. In addition, d-separation satisfies certain properties that are in general not obeyed by abstract conditional independence, including the following \citep[Theorem 11]{pearl1988book}, where $W,X,Y,Z$ are disjoint sets of vertices of a directed acyclic graph $\g$.

\begin{description}
\item[Composition] $X \indep_{\g} Y \mid Z$ and $X \indep_{\g} W \mid Z$  $\implies X \indep_{\g} W,Y \mid Z$.
\item[Weak transitivity] $X \indep_{\g} Y \mid Z$ and $X \indep_{\g} Y \mid Z, s$ for vertex $s$ $\implies$ either $s \indep_{\g} X \mid Z$ or $s \indep_{\g} Y \mid Z$.
\end{description}
We strengthen weak transitivity in \cref{thm:transitivity}.

\smallskip DAG $\g$ defines the Bayesian network model over the variables $V$ in the graph. We say a distribution $P$ follows the Bayesian network model, or $P$ is Markov to $\g$, if $P$ factorizes according to the graph:
\[ p(V) = \prod_{v} p(v \mid \Pa_{\g}(v)),\]
where $p$ is the density of $P$ with respect to some product dominating measure. It can be shown that $P$ is Markov to $\g$ if and only if $P$ obeys the global Markov property implied by $\g$: for disjoint subsets $X,Y,Z$ of $V$,
\[ X \indep_{\g} Y \mid Z \:\implies\: X \indep Y \mid Z \text{ under $P$}. \]
See, e.g., \citet[Theorem 3.27]{lauritzen1996graphical}. If the reverse holds, i.e., for  disjoint subsets $X,Y,Z$,
\[ X \indep Y \mid Z \text{ under $P$} \:\implies\: X \indep_{\g} Y \mid Z, \]
we say $P$ is faithful to $\g$.

\subsection{Causal model and SWIGs} \label{apx:swig}
A causal DAG $\g$ can be associated with different assumptions on the
distribution of factual and counterfactual random variables, depending
on how we interpret graph $\g$
\citep{robins11_alter_graph_causal_model_ident_direc_effec}. Most
notably, \citet{robins1986new}
introduced the ``finest fully randomized causally interpretable
structured tree graph'' (FFRCISTG) model, while
\citet{pearl1995causal} interpreted $\g$ as positing a non-parametric
structural equation model with independent errors (NPSEM-IE). The
conditional independencies implied by the FFRCISTG model can
be read off from the corresponding single-world intervention graph
(SWIG) using d-separation \citep{richardson2013single}. For the same graph $\g$, the NPSEM-IE model is a submodel of the
FFRCISTG model, as the former additionally posits cross-world
independencies \citep{richardson2013single,shpitser2022multivariate}.

\subsection{Conjunctive cause criterion} \label{apx:common-cause}
We prove $A \indep Y(a) \mid \Ccommon$ holds under \cref{assump:S=Z} for the FFRCISTG potential outcome model (and hence the stronger NPSEM-IE model) represented by graph $\g$.

\begin{proof}[Proof of \cref{prop:common-cause}]
  Let $\g(a)$ be the SWIG corresponding to intervening on $A$ and imposing value $a$. Graph $\g(a)$ is formed from $\g$ by splitting $A$ into a random part $A$ and a fixed part $a$, where $A$ inherits all the edges into $A$ and $a$ inherits all the edges out of $A$. Additionally, any descendant $V_i$ of $A$ is labelled as $V_i(a)$ in $\g(a)$. In particular, $Y$ is labelled as $Y(a)$ in $\g(a)$. To prove our result, it suffices to show that $A$ and $Y(a)$ are d-separated in $\g(a)$ given $\Ccommon$.

Under $Z=S$, by the fact that $\g$ is a DAG over $Z \cup \{A,Y\}$, we have $\Ccommon = \An_{\g}(A) \cap \An_{\g}(Y) \setminus \{A\}$. Suppose there is a path $p$ in $\g(a)$ that d-connects $A$ and $Y(a)$ given $\Ccommon$. First, observe that this is impossible if $p$ does not contain any collider, since by construction of $\g(a)$ (no edge stems out of $A$) and $Y \notin \An_{\g}(A)$, $p$ must be of the form $A \leftarrow \dots \rightarrow Y(a)$ and is thus blocked by $\Ccommon$. Hence, $p$ must contain at least one collider. Further, again by construction of $\g(a)$, $p$ must also contain non-colliders. Let $\gamma$ be the collider on $p$ that is closest to $A$. Also let $\delta$ be the vertex that precedes $\gamma$ on the subpath $p(A,\gamma)$.
Vertex $\delta \neq A$ is a non-collider on $p$. For $p$ to d-connect, $\gamma$ is an ancestor of some $v \in \Ccommon$. However, this implies that $\delta \in \Ccommon$ and hence $p$ is blocked. Therefore, such a d-connecting path $p$ cannot exist.
\end{proof}

\subsection{Disjunctive cause criterion} \label{apx:disj-cause}
The following proof of \cref{prop:disj-cause} is due to \citet{richardson2018discussion}, which is based on the following result on inducing paths.

\begin{definition} \label{def:inducing}
Consider vertices $u, v$ and a vertex set $L \subset V \setminus \{u,v\}$. A path between $u$ and $v$ is called an inducing path relative to $L$ if every non-endpoint vertex on the path that is not in $L$ is (1) a collider and (2) is an ancestor of $u$ or $v$.
\end{definition}

\begin{lemma}[\citet{verma1990equivalence}] \label{lem:inducing}
Let $\g$ be a DAG over vertices $V$. Fix two vertices $u,v$ and a set $L \subset V$ such that $u, v \notin L$. Then $u$ and $v$ cannot be d-separated by any subset of $V \setminus (L \cup \{u,v\})$ if and only if there exists an inducing path between $u$ and $v$ relative to $L$ in $\g$.
\end{lemma}

\begin{proof}[Proof of \cref{prop:disj-cause}]
Let DAG $\g$ over vertices $V$ be the underlying causal graph. We prove the statement by contradiction. Suppose $\CvdW$ is not a sufficient adjustment set. By \cref{prop:backdoor}, there exists some d-connecting backdoor path $\pi$ between $A$ and $Y$ given $S \cap (\An_{\g}(A) \cup \An_{\g}(Y))$. Remove edges out of $A$ from $\g$ and call the resulting graph $\tilde{\g}$. Observe that $\An_{\tilde{\g}}(A) \cup \An_{\tilde{\g}}(Y) = \An_{\g}(A) \cup \An_{\g}(Y)$, and the paths between $A$ and $Y$ in $\tilde{\g}$ exactly correspond to the back-door paths between $A$ and $Y$ in $\g$. Moreover, the path $\pi$ d-connects $A$ and $Y$ in $\tilde{\g}$ given $S \cap (\An_{\tilde{\g}}(A) \cup \An_{\tilde{\g}}(Y))$.

We claim that $\pi$ is an inducing path between $A$ and $Y$ in $\tilde{\g}$ relative to $L \equiv V \setminus (S \cup \{A,Y\})$. To see this, let $k$ be any non-endpoint vertex on $\pi$. Suppose $k$ is a non-collider. Because $\pi$ is d-connected given $S \cap (\An_{\tilde{\g}}(A) \cup \An_{\tilde{\g}}(Y))$, it follows from \cref{lem:an-d-conn} that $k$ must be either an ancestor of $A$, an ancestor of $Y$, or an ancestor of $S \cap (\An_{\tilde{\g}}(A) \cup \An_{\tilde{\g}}(Y))$. However, in any case, $k$ is included in the conditioning set and $\pi$ would be blocked. Hence, $k$ is a collider. Further, by the d-connection of $\pi$, $k$ must be an ancestor of $S \cap (\An_{\tilde{\g}}(A) \cup \An_{\tilde{\g}}(Y))$, which implies $k \in \An_{\tilde{\g}}(A) \cup \An_{\tilde{\g}}(Y)$. Hence, $\pi$ is an inducing path between $A$ and $Y$ relative to $L$. By \cref{lem:inducing}, $A$ and $Y$ cannot be d-separated in $\tilde{\g}$ by any subset of $S$. This contradicts the existence of a sufficient adjustment set postulated by \cref{assump:some-subset-S}, which by \cref{prop:backdoor} is equivalent to the existence of a subset $S$ blocking all back-door paths from $A$ to $Y$ under the faithfulness assumption.
\end{proof}

\subsection{Minimality of \texorpdfstring{$\CAY^{\ast}$}{C\_{AY}\^*} and \texorpdfstring{$\CYA^{\ast}$}{C\_{YA}\^*} } \label{apx:minimal}
We first strengthen the weak transitivity property of d-separation \citep[cf][Theorem 12]{pearl1988book}.

\begin{theorem}[Weak transitivity, strengthened] \label{thm:transitivity}
Let $\g$ be a DAG. Let $x,y$ be two vertices and $W, Z$ be two disjoint vertex sets such that $x,y \notin W \cup Z$. Suppose $x \indep_{\g} y \mid W$ and $x \indep_{\g} y \mid W, Z$. Then there exists $z \in Z$ such that either $z \indep_{\g} x \mid W$ or $z \indep_{\g} y \mid W$ holds.
\end{theorem}

To prove \cref{thm:transitivity}, we need the following two lemmas on d-connecting paths.
\begin{lemma} \label{lem:an-d-conn}
Let $\pi$ be a d-connecting path between $x$ and $y$ given $W$. Then, every vertex on $\pi$ is in $\An(\{x,y\} \cup W)$. In consequence, if $z$ on the path is not an ancestor of $W$, then either $\pi(z,x)$ or $\pi(z,y)$ is a causal path from $z$ to the endpoint.
\end{lemma}
\begin{proof}
It follows from the fact that every vertex on a d-connecting path is an ancestor of either endpoint or a collider.
\end{proof}

\begin{lemma} \label{lem:concat}
Let $\g$ be a DAG. Suppose path $\pi_1$ d-connects $x$ and $z$ given $W$ and path $\pi_2$ d-connects $y$ and $z$ given $W$. Suppose $z \in \An_{\g}(W)$. Then $x$ and $y$ are d-connected given $W$.
\end{lemma}
\begin{proof}
Let $t$ be the first vertex after $x$ on $\pi_1$ that is also on $\pi_2$. Let $\pi^{\ast}$ be the path formed by concatenating $\pi_1(x,t)$ and $\pi_2(t,y)$ (one can check that there are no duplicating vertices in $\pi^{\ast}$). Observe that every non-endpoint vertex on $\pi^{\ast}$ has the same collider status as the same vertex on $\pi_1$ or $\pi_2$, with possible exception of $t$. Thus, it suffices to show that $\pi^{\ast}$ is not blocked by $t$.

Suppose $t \neq y$, otherwise $\pi^{\ast}$ is obviously d-connected. There are two scenarios:
\begin{enumerate}[(i)]
\item The vertex $t$ is a collider on $\pi^{\ast}$; it suffices to show $t \in \An_{\g}(W)$. This is immediately true if $t = z$. Otherwise, suppose $t \neq z$ and $t \notin \An_{\g}(W)$. Then by \cref{lem:an-d-conn}, that $\pi(x,t)$ terminates with an arrow into $t$ implies that $\pi_1(t,z)$ must be of the form $t \rightarrow \dots \rightarrow z$, which gives $t \in \An_{\g}(W)$ and hence a contradiction.
\item The vertex $t$ is a non-collider on $\pi^{\ast}$. it suffices to show $t \notin W$. This is immediately true if $t \neq z$. Otherwise, $t$ must be a non-collider on either $\pi_1$ or $\pi_2$ (in order for $t$ to be a non-collider on the concatenated path $\pi^*$). The d-connectedness of $\pi_1$ and $\pi_2$ also implies $t \notin W$.
\end{enumerate}
\end{proof}

\begin{proof}[Proof of \cref{thm:transitivity}]
We prove the statement by contradiction. Suppose for every $z \in Z$, there exist path $p_z$ that d-connects $x$ and $z$ given $W$ and path $q_z$ that d-connects $y$ and $z$ given $W$. Observe that $p_z$ cannot be a causal path from $z$ to $x$ and $q_z$ cannot be a causal path from $z$ to $y$, since otherwise by concatenating the two paths at the first vertex they intersect (which is a non-collider), $x$ and $y$ are d-connected given $W$, which contradicts our assumption.

We claim that $Z \cap \An_{\g}(W) = \emptyset$. Otherwise, suppose $z' \in Z$ is an ancestor of $W$. Then by \cref{lem:concat} and the existence of paths $p_{z'}, q_{z'}$, vertices $x$ and $y$ are d-connected given $W$, which contradicts $x \indep_{\g} y \mid W$.

Choose $z^0$ from $Z$ such that no other element of $Z$ is an ancestor of $z^0$. Consider path $p_{z^0}$ between $x,z^0$ and path $q_{z^0}$ between $y,z^0$. Both paths are d-connecting given $W$. Let $t$ be first vertex after $x$ on $p_{z^0}$ that is also on $q_{z^0}$, with $t=z^0$ as a special case. Let $\pi$ be the path formed by concatenating $p_{z^0}(x,t)$ and $q_{z^0}(t,y)$. There are two cases.
\begin{enumerate}[(i)]
\item Vertex $t$ is a collider on $\pi$. By $x \indep_{\g} y \mid W$ and \cref{lem:concat}, $t$ is not an ancestor of $W$. Then by \cref{lem:an-d-conn}, $p_{z^0}(t,z^0)$ must be of the form $t \rightarrow \dots \rightarrow z^0$. We claim that $\pi$ between $x$ and $y$ is d-connecting given $Z$ and $W$, which would contradict $x \indep_{\g} y \mid Z, W$. To see this, let us inspect every non-endpoint vertex $v$ on $\pi$ and verify the condition on $v$ for $\pi$ to d-connect given $W,Z$. If $v = t$, then $v$ is a collider and $v$ is an ancestor of $Z$ and hence of $W \cup Z$.

Now suppose $v \neq t$. The status (collider or non-collider) of $v$ is the same as its status on $p_{z^0}$ or $q_{z^0}$.
We show $v$ does not block $\pi$. If $v$ is a collider, then it is an ancestor of $W$ and hence of $W \cup Z$. Otherwise, we know $v \notin W$ and it suffices to show $v \notin Z$. For a contradiction, suppose $v \in Z$ is a non-collider on $\pi$. As observed in the beginning of the proof, $\pi(v,x)$ and $\pi(v,y)$ cannot be causal paths from $v$ to the other endpoint. It follows that on the subpath $\pi(x,t)$ or $\pi(y,t)$ that contains $v$, vertex $v$ must be an ancestor of a collider (and hence of $W$) or an ancestor of $t$ (and hence of $z^0$), neither of which is possible: we showed $Z \cap \An_{\g}(W) = \emptyset$ and we choose $z^0$ such that no other vertex in $Z$ is an ancestor of $z^0$. Hence, $\pi$ d-connects $x$ and $y$ given $Z$ and $W$, which contradicts our assumption $x \indep_{\g} y \mid Z, W$.

\item Vertex $t$ is a non-collider on $\pi$. Then $t$ is a non-collider on either $p_{z^0}$ or $q_{z^0}$, which implies $t \notin W$. Because the status of any other non-endpoint vertex on $\pi$ remains the same as on $p_{z^0}$ or $q_{z^0}$, we know $\pi$ d-connects $x$ and $y$ given $W$, which contradicts our assumption $x \indep_{\g} y \mid W$.
\end{enumerate}
\end{proof}

We are ready to prove that $\CAY^{\ast}(S)$ and $\CYA^{\ast}(S)$ are minimal sufficient adjustment sets.
\begin{proof}[Proof of \cref{thm:minimal}]
That $C^{\ast}$ is a sufficient adjustment set follows from iteratively applying \cref{lem:predA,lem:predYA} under \cref{assump:S-already,assump:positivity}.

Now under the additional assumption of a causal graph $\g$ and the faithfulness with respect to SWIG $\g(a)$, we show no proper subset of $C^{\ast}$ is a sufficient adjustment set. Suppose $C' \subsetneq C^{\ast}$ controls for confounding. By faithfulness, we have
\begin{equation} \label{eqs:sep-C'}
 A \indep_{\g(a)} Y_a  \mid C'
\end{equation}
and
\begin{equation} \label{eqs:sep-C*}
 A \indep_{\g(a)} Y_a  \mid C^{\ast}.
\end{equation}
By \cref{thm:transitivity}, we know there exists $s \in (C^{\ast} \setminus C')$ such that either (or both) of the following holds.
\begin{enumerate}[(i)]
\item $A \indep_{\g(a)} s \mid C'$. Given also \cref{eqs:sep-C'}, by composition of d-separation (see \cref{apx:d-sep}), we have $A \indep_{\g(a)} s, Y_a \mid C'$. By weak union, it follows that $C'' = C' \cup \{s\}$ also controls for confounding.

\item $Y_a \indep_{\g(a)} s \mid C'$. Given also \cref{eqs:sep-C'}, by composition of d-separation, we have $Y_a \indep_{\g(a)} s, A \mid C'$. By weak union, it follows that $C'' = C' \cup \{s\}$ also controls for confounding.
\end{enumerate}
In either case, we see that $C'$ can be replaced by a set $C''$ whose size is increased by one. Iterating this argument until we get $\tilde{C} = C^{\ast} \setminus \{\tilde{s}\}$ for some element $\tilde{s} \in C^{\ast}$. We have
\begin{equation} \label{eqs:sep-Ctilde}
A \indep_{\g(a)} Y_a  \mid \tilde{C}.
\end{equation}
Applying \cref{thm:transitivity} again to \cref{eqs:sep-Ctilde,eqs:sep-C*}, at least one of the following holds.
\begin{enumerate}[(i)]
\item $A \indep_{\g(a)} C^{\ast} \setminus \tilde{C} \mid \tilde{C}$. By \cref{lem:reduce}, this means the treatment Markov boundary of $C^{\ast}$ is contained in $\tilde{C}$, contradicting \cref{lem:stability}.
\item $Y_a \indep_{\g(a)} C^{\ast} \setminus \tilde{C} \mid \tilde{C}$. Combining it with \cref{eqs:sep-Ctilde} and applying composition, we have $Y_a \indep_{\g(a)} A, C^{\ast} \setminus \tilde{C} \mid \tilde{C}$, which by weak union implies
\[Y_a \indep C^{\ast} \setminus \tilde{C} \mid A=a, \, \tilde{C}.\]
Using consistency, we have
\[ Y \indep C^{\ast} \setminus \tilde{C} \mid A=a, \, \tilde{C}, \quad a=0,1. \]
By \cref{lem:reduce}, this means the outcome Markov boundary of $C^{\ast}$ is contained in $\tilde{C}$, which again contradicts \cref{lem:stability}.
\end{enumerate}
\end{proof}

\subsection{Causal closure} \label{apx:closure}
In this subsection, we discuss the notion of causal closure with respect to a fixed ground DAG $\bar{\g}$, which is the underlying causal DAG that includes every variable in the system.
For two distinct vertices $v_1, v_2$ on $\bar{\g}$, we define their non-trivial common ancestors to be
\begin{multline*}
\An_{\bar{\g}}^{\ast}(v_1, v_2) \equiv \big\{u \in \An_{\bar{\g}}(v_1) \cap \An_{\bar{\g}}(v_2): \text{$u$ has a causal}\\ \text{path to $v_1$ or $v_2$ not through $\An_{\bar{\g}}(v_1)  \cap \An_{\bar{\g}}(v_2) \setminus \{u\}$} \big\}.
\end{multline*}
For example, if $\bar{\g}$ is $Z \rightarrow X \rightarrow Y$, then $Z$ is a common ancestor of $(X,Y)$ but not a non-trivial common ancestor: $Z$ causally goes to $X$ and $Y$ only through $\An(X) \cap \An(Y) \setminus \{Z\} = \{X\}$.

\begin{definition}[causal closure] \label{def:closure}
A vertex set $H$ is \emph{causally closed} if
\[ v_1, v_2 \in H,\, v_1 \neq v_2 \quad \implies \quad \An_{\bar{\g}}^{\ast}(v_1, v_2) \subset H. \]
Further, the \emph{causal closure} of $H$ is
\[ \overline{H} \equiv \bigcap\, \{H' \supseteq H: \text{$H'$ is causally closed} \}. \]
\end{definition}
Observe that a causal closure is always causally closed. Causal closure is defined with respect to those non-trivial common ancestors so pre-treatment variables irrelevant to the effect need not be included. In the example of \cref{fig:closure}, the causal closure of $\{S_1,S_2,A,Y\}$ is $\{Z_1,Z_3,S_1,S_2,A,Y\}$.

\begin{figure}[!htb]
  \centering
  \begin{tikzpicture}
    \tikzset{rv/.style={circle,inner sep=1pt,fill=gray!20,draw,font=\sffamily},
      fv/.style={circle,inner sep=1pt,fill=gray!40,draw,font=\sffamily},
      node distance=12mm, >=stealth}
      \node[rv] (A) {$A$};
      \node[rv, right=15mm of A] (Y) {$Y$};
      \node[rv, above of=A, xshift=3mm] (S1) {$S_1$};
      \node[rv, above of=Y, xshift=-3mm] (S2) {$S_2$};
      \node[rv, above of=S1, xshift=6mm] (Z1) {$Z_1$};
      \node[rv, left of=A] (Z2) {$Z_2$};
      \node[rv, left of=Z1](Z3) {$Z_3$};
      \node[rv, left of=S1](Z4) {$Z_4$};
      \draw[->,very thick] (A) -- (Y);
      \draw[->,very thick] (S1) -- (A);
      \draw[->,very thick] (S1) -- (Y);
      \draw[->,very thick] (S2) -- (A);
      \draw[->,very thick] (S2) -- (Y);
      \draw[->,very thick] (Z1) -- (S1);
      \draw[->,very thick] (Z1) -- (S2);
      \draw[->,very thick] (Z2) -- (A);
      \draw[->,very thick] (Z3) -- (Z1);
      \draw[->,very thick] (Z3) -- (Z4);
      \draw[->,very thick] (Z4) -- (S1);
  \end{tikzpicture}
  \caption{A ground DAG $\bar{\g}$}
  \label{fig:closure}
\end{figure}

\begin{lemma} \label{lem:closure}
For $u \in \overline{H}$, either $u \in H$ or $u$ is a common ancestor of two distinct vertices in $H$.
\end{lemma}
\begin{proof}
We prove the statement by contradiction. Suppose $u \in \overline{H} \setminus H$ and $\De(u) \cap H$ is empty or a singleton. Consider $H' = \overline{H} \setminus (\De(u) \setminus H)$. By construction, $H \subseteq H' \subset \overline{H}$. Further, observe that $H'$ is causally closed. However, this contradicts the definition of $\overline{H}$.
\end{proof}

\begin{proposition} \label{prop:Z-pre}
The set $Z$ given by \cref{eqs:Z-by-closure} is pre-treatment.
\end{proposition}
\begin{proof}
\cref{eqs:Z-by-closure} states that
\[ Z = \overline{S \cup \{A,Y\}} \setminus \{A,Y\}. \]
Because $S$ is pre-treatment, we shall show every $z \in Z \setminus S$ is pre-treatment.
By \cref{lem:closure}, $z$ is a common ancestor to two distinct vertices $v_1, v_2 \in S \cup \{A,Y\}$. It follows that either $v_1$ or $v_2$ is in $S \cup \{A\}$. It is clear that $z$ is an ancestor of $S \cup \{A\}$ and hence pre-treatment.
\end{proof}

\begin{proposition} \label{prop:Z-adj}
Suppose $\bar{\g}$ is a DAG in which $A \in \An_{\bar{\g}}(Y)$. Let $Z$ be a vertex set that contains no descendant of $A$. If $Z \cup \{A,Y\}$ is causally closed, then $Z$ is a sufficient adjustment set.
\end{proposition}
\begin{proof}
We prove by contradiction. Suppose $Z$ is not a sufficient adjustment set. Then, by \cref{prop:backdoor}, there exists a back-door path $\pi$ between $A$ and $Y$ that is d-connected given $Z$. There are two cases. If $\pi$ contains no collider, then one can show that there exists a non-collider vertex $L$ on $\pi$ such that $L \in \An_{\bar{\g}}^{\ast}(A,Y)$; otherwise, one can show that there exists a non-collider vertex $L$ on $\pi$ such that $L \in \An_{\bar{\g}}^{\ast}(A,z) \cup \An_{\bar{\g}}^{\ast}(z,Y)$ for some $z \in Z$. In either case, because $\pi$ d-connects $A$ and $Y$ given $Z$, it holds that $L \notin Z \cup \{A,Y\}$. However, this contradicts our assumption that $Z \cup \{A,Y\}$ is causally closed.
\end{proof}

\bibliographystyle{imsart-nameyear}

\end{document}